\documentclass[twoside,leqno,twocolumn]{article}
\pdfoutput=1

\usepackage[english]{babel}
\usepackage[T1]{fontenc}
\usepackage[utf8]{inputenc}

\usepackage[letterpaper]{geometry}

\usepackage{ltexpprt}

\usepackage{algorithm}
\usepackage{algpseudocode}
\usepackage{amsfonts}
\usepackage{amsmath}
\usepackage{balance}
\usepackage{bm}
\usepackage{booktabs}
\usepackage{cite}
\usepackage{comment}
\usepackage{mathtools}
\usepackage{nicefrac}
\usepackage{siunitx}
\usepackage{tabu}
\usepackage{xparse}
\usepackage{xspace}

\usepackage[hidelinks]{hyperref}

\DeclarePairedDelimiter{\ceil}{\lceil}{\rceil}
\DeclarePairedDelimiter{\floor}{\lfloor}{\rfloor}

\NewDocumentCommand{\Math}{m}{%
  \ensuremath{#1}\xspace%
}

\NewDocumentCommand{\N}{o o}{%
  \IfValueTF{#1}{%
    \IfValueTF{#2}{%
      \Math{\mathbb{N}^{#1 \times #2}}%
    }{%
      \Math{\mathbb{N}^{#1}}%
    }%
  }{%
    \Math{\mathbb{N}}%
  }%
}

\NewDocumentCommand{\R}{o o}{%
  \IfValueTF{#1}{%
    \IfValueTF{#2}{%
      \Math{\mathbb{R}^{#1 \times #2}}%
    }{%
      \Math{\mathbb{R}^{#1}}%
    }%
  }{%
    \Math{\mathbb{R}}%
  }%
}

\NewDocumentCommand{\var}{o o d() m}{%
  \IfNoValueTF{#1}{%
    \IfNoValueTF{#3}{%
      \Math{#4}%
    }{%
      \Math{#4^{(#3)}}%
    }%
  }{%
    \IfNoValueTF{#2}{%
      \IfNoValueTF{#3}{%
        \Math{#4_{#1}}%
      }{%
        \Math{#4_{#1}^{(#3)}}%
      }%
    }{%
      \IfNoValueTF{#3}{%
        \Math{#4_{#1#2}}%
      }{%
        \Math{#4_{#1#2}^{(#3)}}%
      }%
    }%
  }%
}

\RenewDocumentCommand{\vec}{o d() m}{%
  \var[#1](#2){\bm{#3}}%
}

\NewDocumentCommand{\vectilde}{o d() m}{%
  \var[#1](#2){\tilde{\bm{#3}}}%
}

\NewDocumentCommand{\mat}{o o m}{%
  \var[#1][#2]{\bm{#3}}%
}

\NewDocumentCommand{\matinv}{o o m}{%
  \var[#1][#2]{\bm{#3}^{-1}}%
}

\NewDocumentCommand{\vecconcat}{m m}{%
  \Math{[#1,#2]}%
}

\NewDocumentCommand{\scalprod}{m m}{%
  \Math{(#1,#2)}%
}

\NewDocumentCommand{\scalalpha}{d()}{%
  \var(#1){\alpha}%
}

\NewDocumentCommand{\scalbeta}{d()}{%
  \var(#1){\beta}%
}

\NewDocumentCommand{\scalgamma}{d()}{%
  \var(#1){\gamma}%
}

\NewDocumentCommand{\scaldelta}{d()}{%
  \var(#1){\delta}%
}

\NewDocumentCommand{\scalzeta}{d()}{%
  \var(#1){\zeta}%
}

\NewDocumentCommand{\scaleta}{d()}{%
  \var(#1){\eta}%
}

\NewDocumentCommand{\scaltheta}{d()}{%
  \var(#1){\theta}%
}

\NewDocumentCommand{\scalkappa}{o d()}{%
  \var[#1](#2){\kappa}%
}

\NewDocumentCommand{\scallambda}{o d()}{%
  \var[#1](#2){\lambda}%
}

\NewDocumentCommand{\vecb}{o}{%
  \vec[#1]{b}%
}

\NewDocumentCommand{\vecc}{o d()}{%
  \vec[#1](#2){c}%
}

\NewDocumentCommand{\vecd}{o d()}{%
  \vec[#1](#2){d}%
}

\NewDocumentCommand{\vecg}{o d()}{%
  \vec[#1](#2){g}%
}

\NewDocumentCommand{\vech}{o d()}{%
  \vec[#1](#2){h}%
}

\NewDocumentCommand{\vecm}{o d()}{%
  \vec[#1](#2){m}%
}

\NewDocumentCommand{\vecn}{o d()}{%
  \vec[#1](#2){n}%
}

\NewDocumentCommand{\vecp}{o d()}{%
  \vec[#1](#2){p}%
}

\NewDocumentCommand{\vecq}{o d()}{%
  \vec[#1](#2){q}%
}

\NewDocumentCommand{\vecr}{o d()}{%
  \vec[#1](#2){r}%
}

\NewDocumentCommand{\vecs}{o d()}{%
  \vec[#1](#2){s}%
}

\NewDocumentCommand{\vecu}{o d()}{%
  \vec[#1](#2){u}%
}

\NewDocumentCommand{\vecw}{o d()}{%
  \vec[#1](#2){w}%
}

\NewDocumentCommand{\vecx}{o d()}{%
  \vec[#1](#2){x}%
}

\NewDocumentCommand{\vecz}{o d()}{%
  \vec[#1](#2){z}%
}

\NewDocumentCommand{\matA}{o o}{%
  \mat[#1][#2]{A}%
}

\NewDocumentCommand{\matM}{o o}{%
  \mat[#1][#2]{M}%
}

\NewDocumentCommand{\matMinv}{o o}{%
  \matinv[#1][#2]{M}%
}

\NewDocumentCommand{\matP}{o o}{%
  \mat[#1][#2]{P}%
}

\NewDocumentCommand{\Ref}{m m m}{%
  #1~\ref{#2:#3}\xspace%
}

\NewDocumentCommand{\Refs}{m m m m}{%
  #1~\ref{#2:#3} and~\ref{#2:#4}\xspace%
}

\NewDocumentCommand{\REFS}{m m m m}{%
  #1~\ref{#2:#3} to~\ref{#2:#4}\xspace%
}

\NewDocumentCommand{\Alg}{m}{%
  \Ref{Algorithm}{alg}{#1}%
}

\NewDocumentCommand{\Algs}{m m}{%
  \Refs{Algorithms}{alg}{#1}{#2}%
}

\NewDocumentCommand{\Eqn}{m}{%
  \Ref{Equation}{eqn}{#1}%
}

\NewDocumentCommand{\Figs}{m m}{%
  \Refs{Figures}{fig}{#1}{#2}%
}

\NewDocumentCommand{\Lem}{m}{%
  \Ref{Lemma}{lem}{#1}%
}

\NewDocumentCommand{\Lin}{m m}{%
  \Ref{line}{alg}{#1:#2} of \Alg{#1}%
}

\NewDocumentCommand{\Lins}{m m m}{%
  \Refs{lines}{alg}{#1:#2}{#1:#3} of \Alg{#1}%
}

\NewDocumentCommand{\LINS}{m m m}{%
  \REFS{lines}{alg}{#1:#2}{#1:#3} of \Alg{#1}%
}

\NewDocumentCommand{\Sec}{m}{%
  \Ref{Section}{sec}{#1}%
}

\NewDocumentCommand{\Secs}{m m}{%
  \Refs{Sections}{sec}{#1}{#2}%
}

\NewDocumentCommand{\Tab}{m}{%
  \Ref{Table}{tab}{#1}%
}

\NewDocumentCommand{\all}{}{%
  \Math{\star}%
}

\NewDocumentCommand{\compl}{m}{%
  \Math{\bar{#1}}%
}

\NewDocumentCommand{\nn}{}{%
  \Math{N}%
}

\NewDocumentCommand{\rn}{}{%
  \Math{\rho}%
}

\NewDocumentCommand{\nredu}{}{%
  \Math{\phi}%
}

\NewDocumentCommand{\nfail}{}{%
  \Math{\psi}%
}

\NewDocumentCommand{\nfi}{}{%
  \Math{i}%
}

\NewDocumentCommand{\BiState}{m m}{%
  \State{#1,\;#2}%
}

\NewDocumentCommand{\TriState}{m m m}{%
  \State{#1,\;#2,\;#3}%
}

\NewDocumentCommand{\latency}{}{%
  \ensuremath{\mu}\xspace%
}

\NewDocumentCommand{\invbw}{}{%
  \ensuremath{\nu}\xspace%
}

\NewDocumentCommand{\cg}{}{%
  CG\xspace%
}

\NewDocumentCommand{\CR}{}{%
  CR\xspace%
}

\NewDocumentCommand{\pcg}{}{%
  PCG\xspace%
}

\NewDocumentCommand{\pcr}{}{%
  PCR\xspace%
}

\NewDocumentCommand{\ppcg}{}{%
  PPCG\xspace%
}

\NewDocumentCommand{\ppcr}{}{%
  PPCR\xspace%
}

\NewDocumentCommand{\tppcg}{}{%
  2PPCG\xspace%
}

\begin{document}

\title{\Large Scalable Resilience Against Node Failures for
  Communication-Hiding Preconditioned Conjugate Gradient and Conjugate Residual
  Methods}
\author{Markus Levonyak\thanks{University of Vienna, Faculty of Computer
  Science, Vienna, Austria} 
  \and Christina Pacher\footnotemark[1]
  \and Wilfried N.\ Gansterer\footnotemark[1]%
    \protect\phantom{\footnotesize 1}\textsuperscript{,}%
    \thanks{Corresponding author}}
\date{}

\maketitle

\fancyfoot[R]{\scriptsize{Copyright \textcopyright\ 2020 by SIAM\\
Unauthorized reproduction of this article is prohibited}}

\begin{abstract}
\small\baselineskip=9pt
The observed and expected continued growth in the number of nodes in
large-scale parallel computers gives rise to two major challenges: global
communication operations are becoming major bottlenecks due to their limited
scalability, and the likelihood of node failures is increasing.
We study an approach for addressing these challenges in the context of solving
large sparse linear systems.
In particular, we focus on the pipelined preconditioned conjugate gradient
(\ppcg) method, which has been shown to successfully deal with the first of
these challenges.
In this paper, we address the second challenge.
We present extensions to the \ppcg solver and two of its variants which make
them resilient against the failure of a compute node while fully preserving
their communication-hiding properties and thus their scalability.
The basic idea is to efficiently communicate a few redundant copies of local
vector elements to neighboring nodes with very little overhead.
In case a node fails, these redundant copies are gathered at a replacement
node, which can then accurately reconstruct the lost parts of the solver's
state.
After that, the parallel solver can continue as in the failure-free scenario.
Experimental evaluations of our approach illustrate on average very low runtime
overheads compared to the standard non-resilient algorithms.
This shows that scalable algorithmic resilience can be achieved at low extra
cost.

\end{abstract}

\section{Introduction}
\label{sec:introduction}

The \emph{conjugate gradient}~(\cg) and \emph{conjugate residual}~(\CR)
algorithms as well as their preconditioned variants~(\pcg and \pcr) are widely
used iterative Krylov subspace methods for solving linear systems \(\matA \vecx
= \vecb\)~\cite{Hestenes1952a, Saad2003a}.
In many scientific applications, these linear systems are obtained from
discretizing partial differential equations that model the simulated problems.
The resulting system matrix \matA then typically is sparse and may contain only
a very small number of non-zero elements per row.
Both solvers are frequently run in parallel and are expected to be viable
choices for upcoming large-scale parallel computers with hundreds of thousand
or even millions of compute nodes.

However, two major challenges have to be overcome in order to optimize the PCG
and PCR methods for such large parallel computers.
Firstly, communication between a substantial fraction or even all of the
compute nodes, i.e., \emph{global} communication, becomes increasingly
expensive with a growing total number \nn of nodes.
This is particularly true for computing the dot products in each iteration of
PCG and PCR, which involves costly global synchronization.
The cost of a global reduction operation is \(\mathcal{O}(\log \nn)\) and,
thus, steadily grows with an increasing number of nodes~\cite{Ghysels2014a,
Hoefler2010a}.
Secondly, the reliability of computer clusters is predicted to deteriorate at
scale.
A compute node of a cluster may fail for many different reasons, e.g., some
hardware component malfunctions, the shared memory gets corrupted, or it loses
its connection to the interconnection network.
If we assume a---rather optimistic---mean time between failures~(MTBF) of a
century for an individual node, a cluster with \num{e5}~nodes, on average, will
encounter a node failure every nine hours.
Even worse, a system with \num{e6}~nodes will encounter a node failure every
53~minutes on average~\cite{Herault2015a}.
Therefore, we have to expect possibly several node failures during the
execution of a long-running scientific application.

To tackle the first challenge, solvers that try to either reduce global
communication or overlap global communication with both computation and local
communication have been suggested.
These two classes of solvers are commonly referred to as
\emph{communication-avoiding} and \emph{communication-hiding} solvers,
respectively.
Early communication-avoiding \cg methods comprise variants of the original
algorithm with only a single global synchronization point~\cite{Barrett1994a,
DAzevedo1992a, DAzevedo1993a, Eijkhout1992a, Meurant1987a, Saad1984a} as well
as a \cg variation with two three-term recurrences and only one reduction
operation~\cite{Saad2003a}.
For reducing global communication even further, \(s\)-step methods have been
introduced~\cite{Carson2018a, Chronopoulos1989a, Hoemmen2010a} and recently
applied in large-scale simulations~\cite{Idomura2018a, Mayumi2016a}.
These methods accomplish to reduce the number of global synchronizations by
\(\mathcal{O}(s)\) through computing the \cg iterations in blocks of \(s\).
However, \(s\)-step methods tend to become numerically unstable with increasing
\(s\).
Another approach for reducing global communication has been to enlarge the
Krylov subspace~\cite{Grigori2016a}.

Besides the overall reduction of global synchronization points, it has early
been suggested to overlap the dot products in the \cg method with
computation~\cite{DeSturler1995a, Demmel1993a}.
This principle has been enhanced by Ghysels et al.~\cite{Ghysels2013a,
Ghysels2014a} with the introduction of \emph{pipelined} solvers, first for the
generalized minimal residual~(GMRES) method and later for \pcg.
The pipelined \pcg~(\ppcg) algorithm performs only a single non-blocking
reduction per iteration and overlaps this global communication with the
application of the preconditioner as well as the computation of the sparse
matrix-vector product~(SpMV), which often requires solely local
communication~\cite{Ghysels2014a}.
Based on \ppcg, Ghysels and Vanroose derive the closely related pipelined
\pcr~(\ppcr) algorithm~\cite{Ghysels2014a}.
Both \ppcg and \ppcr are readily available in the widely used parallel
numerical library PETSc~\cite{Balay1997a, Balay2019a}.
Moreover, several studies investigate the numerical properties and performance
of the \ppcg method~\cite{Carson2018b, Cools2016a, Cools2018a, Cools2018b,
Cools2019a, Cools2019b}.
Others propose modified or alternative pipelined \pcg
methods~\cite{Cornelis2018a, Cools2017a, Eller2016a, Sanan2016a}, one of them
being the two-iteration pipelined \pcg~(\tppcg) algorithm by Eller and
Gropp~\cite{Eller2016a}, which is based on a three-term recurrence variant of
\pcg.

To overcome the second challenge described above, different general-purpose and
algorithm-specific fault-tolerance approaches have been discussed in the
literature.
Nowadays, the most commonly applied measures against node failures are various
\emph{checkpointing} and \emph{rollback-recovery} techniques, which frequently
save the full state of an executed application and restore the latest one in
case of a node failure~\cite{Herault2015a, Tiwari2014a}.
To avoid the usually considerable overhead of continuously saving the state of
an entire application, Chen~\cite{Chen2011a} and Pachajoa et
al.~\cite{Pachajoa2018a, Pachajoa2019a} exploit the inherent redundancy of the
SpMV in \pcg.
A well-defined strategy ensures enough redundant copies of the search direction
vectors in order to fully recover the whole state of \pcg after possibly
multiple simultaneous node failures.
An alternative approach by Langou et al.~\cite{Langou2007a} and Agullo et
al.~\cite{Agullo2013a, Agullo2016a} approximates the lost part of the latest
solution vector, which is then used as the initial guess for the restarted
solver.
Bosilca et al.~\cite{Bosilca2014a, Bosilca2015a, Herault2015a} suggest an
algorithm for integrating algorithm-specific with general-purpose
fault-tolerance techniques.
While Pachajoa and Gansterer \cite{Pachajoa2017a} evaluate the inherent
resilience properties of \cg after a node failure, others discuss the related
but independent problem of soft errors in \cg~\cite{Agullo2018a,
Bronevetsky2008a, Dichev2016a, Fasi2016a, Sao2013a, Shantharam2012a}.

In this work, we target the problem of solving large sparse symmetric and
positive-definite~(SPD) linear systems \(\matA \vecx = \vecb\) on parallel
computers that both are susceptible to node failures and have high cost of
global communication compared to computation and local communication.
For this purpose, we introduce an innovative combination of
communication-hiding solvers and algorithm-specific resilience against node
failures.
We focus on the broadly discussed \ppcg solver~\cite{Carson2018b, Cools2016a,
Cools2018a, Cools2018b, Cools2019a, Cools2019b, Ghysels2014a} but also consider
the \ppcr~\cite{Ghysels2014a} and \tppcg~\cite{Eller2016a} algorithms.
For coping with node failures, we propose novel recovery methods for the \ppcg,
\ppcr, and \tppcg solvers, which are partly related to the techniques for the
\pcg method suggested by Chen~\cite{Chen2011a} and Pachajoa et
al.~\cite{Pachajoa2018a, Pachajoa2019a}.
Those techniques are likely to be more \emph{scalable}---i.e., better suited
for large-scale computer clusters---than general-purpose fault-tolerance
techniques.
In numerical experiments, we demonstrate low runtime overheads of our resilient
\ppcg algorithm.

\begin{algorithm}[t]
\caption{Preconditioned conjugate gradient (\pcg)
         method~\cite[Alg.~9.1]{Saad2003a}}
\label{alg:pcg}
\begin{algorithmic}[1]
  \TriState{\(\vecr(0) \gets \vecb - \matA \vecx(0)\)}
           {\(\vecu(0) \gets \matMinv \vecr(0)\)}
           {\(\vecp(0) \gets \vecu(0)\)}
  \State \(\scalgamma(0) \gets \scalprod{\vecr(0)}{\vecu(0)}\)
  \For{\(i \gets 0, 1, \dots, \text{until convergence}\)}
    \State \(\vecs(i) \gets \matA \vecp(i)\)
      \label{alg:pcg:s}
    \State \(\scaldelta(i) \gets \scalprod{\vecs(i)}{\vecp(i)}\)
      \label{alg:pcg:delta}
    \State \(\scalalpha(i) \gets \scalgamma(i) / \scaldelta(i)\)
      \label{alg:pcg:alpha}
    \State \(\vecx(i+1) \gets \vecx(i) + \scalalpha(i) \vecp(i)\)
      \label{alg:pcg:x}
    \State \(\vecr(i+1) \gets \vecr(i) - \scalalpha(i) \vecs(i)\)
      \label{alg:pcg:r}
    \State \(\vecu(i+1) \gets \matMinv \vecr(i+1)\)
      \label{alg:pcg:u}
    \State \(\scalgamma(i+1) \gets \scalprod{\vecr(i+1)}{\vecu(i+1)}\)
      \label{alg:pcg:gamma}
    \State \(\scalbeta(i) \gets \scalgamma(i+1) / \scalgamma(i)\)
      \label{alg:pcg:beta}
    \State \(\vecp(i+1) \gets \vecu(i+1) + \scalbeta(i) \vecp(i)\)
      \label{alg:pcg:p}
  \EndFor
\end{algorithmic}
\end{algorithm}

\subsection{Terminology and assumptions}
\label{sec:introduction:terminology}

As a consequence of a node failure, the affected node becomes unavailable, and
a node that replaces it in the recovery process is called a \emph{replacement
node}.
The replacement node is either a spare node or one of the surviving nodes.
In this paper, we assume that the parallel runtime environment provides
functionality comparable to state-of-the-art implementations of the
industry-standard \emph{Message Passing Interface}~(MPI)~\cite{MPIF2015a}.
Moreover, we assume that the runtime environment provides some basic
fault-tolerance features.
A prototypical example is the \emph{User Level Failure Mitigation}~(ULFM)
framework~\cite{Bland2013a, MPIF2017a}, an extension of the MPI standard.
It supports basic functionality which our approach is based on, including the
detection of node failures, preventing indefinitely blocking synchronizations
or communications, notifying the surviving nodes which nodes have failed, and a
mechanism for providing replacement nodes.

Like in widely used libraries such as PETSc~\cite{Balay1997a, Balay2019a}, we
use a block-row data distribution of all sparse matrices and vectors across the
\nn nodes of the parallel system.
In particular, for an \(n \times n\) linear system, every node owns blocks of
\(n / \nn\) contiguous rows (if \(n = c \nn\) with \(c \in \N\),
otherwise some nodes own \(\floor{n / \nn}\) and others \(\ceil{n / \nn}\)
rows) of all matrices and vectors involved.
On a single node, the data block stored in its shared memory is evenly
distributed among the processors of the node.
Since each node owns rows of all matrices and vectors involved, a node failure
leads to the loss of a part of every matrix and vector.
With the \emph{state} of an iterative solver we mean the---not necessarily
minimal---set of data that completely determines the future behavior of this
iterative solver.

Given a vector \vec(i){v}, where \(i\) denotes the iteration number of the
linear solver, \vec[j](i){v} refers to the subset of elements of the vector at
iteration \(i\) owned by node~\(j\).
For a matrix \mat{B}, the block of rows of \mat{B} owned by node~\(j\) is
denoted by \mat[j][\all]{B}.
On the other hand, \mat[\all][k]{B} is the block of columns of \mat{B}
corresponding to the indices of the rows owned by node~\(k\).
Consequently, \mat[j][k]{B} is the submatrix consisting of the rows owned by
node~\(j\) and the columns corresponding to the indices of the rows owned by
node~\(k\).
\compl{k} stands for all indices except for those of node~\(k\).
\vecconcat{\vec{v}}{\vec{w}} denotes the concatenation of vectors \vec{v} and
\vec{w} to a matrix.
The failed node as well as the replacement node are referred to as node~\rn.

\begin{algorithm}[t]
\caption{Pipelined preconditioned conjugate gradient~(\ppcg)
         method~\cite[Alg.~4]{Ghysels2014a}}
\label{alg:ppcg}
\begin{algorithmic}[1]
  \TriState{\(\vecr(0) \gets \vecb - \matA \vecx(0)\)}
           {\(\vecu(0) \gets \matMinv \vecr(0)\)}
           {\(\vecw(0) \gets \matA \vecu(0)\)}
  \For{\(i \gets 0, 1, \dots, \text{until convergence}\)}
    \BiState{\(\scalgamma(i) \gets \scalprod{\vecr(i)}{\vecu(i)}\)}
            {\(\scaldelta(i) \gets \scalprod{\vecw(i)}{\vecu(i)}\)}
      \label{alg:ppcg:gammadelta}
    \State \(\vecm(i) \gets \matMinv \vecw(i)\)
      \label{alg:ppcg:m}
    \State \(\vecn(i) \gets \matA \vecm(i)\)
      \label{alg:ppcg:n}
    \If{\(i = 0\)}
      \BiState{\(\scalalpha(i) \gets \scalgamma(i) / \scaldelta(i)\)}
              {\(\scalbeta(i) \gets 0\)}
        \label{alg:ppcg:alpha0beta0}
    \Else
      \State \(\scalbeta(i) \gets \scalgamma(i) / \scalgamma(i-1)\)
        \label{alg:ppcg:beta}
      \State \(\scalalpha(i) \gets \scalgamma(i) /
               (\scaldelta(i) - \scalbeta(i) \scalgamma(i) / \scalalpha(i-1))\)
        \label{alg:ppcg:alpha}
    \EndIf
    \State \(\vecz(i) \gets \vecn(i) + \scalbeta(i) \vecz(i-1)\)
      \label{alg:ppcg:z}
    \State \(\vecq(i) \gets \vecm(i) + \scalbeta(i) \vecq(i-1)\)
      \label{alg:ppcg:q}
    \State \(\vecs(i) \gets \vecw(i) + \scalbeta(i) \vecs(i-1)\)
      \label{alg:ppcg:s}
    \State \(\vecp(i) \gets \vecu(i) + \scalbeta(i) \vecp(i-1)\)
      \label{alg:ppcg:p}
    \State \(\vecx(i+1) \gets \vecx(i) + \scalalpha(i) \vecp(i)\)
      \label{alg:ppcg:x}
    \State \(\vecr(i+1) \gets \vecr(i) - \scalalpha(i) \vecs(i)\)
      \label{alg:ppcg:r}
    \State \(\vecu(i+1) \gets \vecu(i) - \scalalpha(i) \vecq(i)\)
      \label{alg:ppcg:u}
    \State \(\vecw(i+1) \gets \vecw(i) - \scalalpha(i) \vecz(i)\)
      \label{alg:ppcg:w}
  \EndFor
\end{algorithmic}
\end{algorithm}

\subsection{Main contributions}
\label{sec:introduction:contributions}

Although communication-hiding (along with communication-avoiding) iterative
solvers and algorithm-specific resilience techniques against node failures are
both motivated by the specific properties of future large-scale parallel
computers, there has been, to the best of our knowledge, no attempt so far to
combine the advantages of those two approaches.
In this paper, we propose novel strategies for recovery after node failures
occurred during the execution of the communication-hiding
\ppcg~\cite{Carson2018b, Cools2016a, Cools2018a, Cools2018b, Cools2019a,
Cools2019b, Ghysels2014a} as well as \ppcr~\cite{Ghysels2014a} and
\tppcg~\cite{Eller2016a} solvers.
To this end, we build upon recent work by Chen~\cite{Chen2011a} and
Pachajoa et al.~\cite{Pachajoa2018a, Pachajoa2019a} regarding resilience
against node failures for the classical \pcg solver.
We eventually show the low runtime overhead of our new fault-tolerant \ppcg
solver in numerical experiments.

\medskip

\noindent The remainder of the paper is structured as follows.
First, in \Sec{solvers}, we discuss in more detail the considered
communication-hiding solvers including their most important properties.
Next, in \Sec{redundancy}, we review how to ensure enough data redundancy for
coping with node failures in the \pcg method and illustrate the relevance for
the \ppcg, \ppcr, and \tppcg algorithms.
Then, in \Sec{recovery}, we derive and describe our novel strategies for
recovering the full state of \ppcg and the other solvers after node failures
occurred.
After that, in \Sec{experiments}, we outline our experiments and present the
results.
Finally, in \Sec{conclusions}, we summarize our conclusions.

\section{Communication-hiding iterative solvers}
\label{sec:solvers}

\begin{algorithm}[t]
\caption{Pipelined preconditioned conjugate residual~(\ppcr)
         method~\cite[Alg.~5]{Ghysels2014a}}
\label{alg:ppcr}
\begin{algorithmic}[1]
  \TriState{\(\vecr(0) \gets \vecb - \matA \vecx(0)\)}
           {\(\vecu(0) \gets \matMinv \vecr(0)\)}
           {\(\vecw(0) \gets \matA \vecu(0)\)}
  \For{\(i \gets 0, 1, \dots, \text{until convergence}\)}
    \State \(\vecm(i) \gets \matMinv \vecw(i)\)
      \label{alg:ppcr:m}
    \BiState{\(\scalgamma(i) \gets \scalprod{\vecw(i)}{\vecu(i)}\)}
            {\(\scaldelta(i) \gets \scalprod{\vecm(i)}{\vecw(i)}\)}
      \label{alg:ppcr:gammadelta}
    \State \(\vecn(i) \gets \matA \vecm(i)\)
      \label{alg:ppcr:n}
    \If{\(i = 0\)}
      \BiState{\(\scalalpha(i) \gets \scalgamma(i) / \scaldelta(i)\)}
              {\(\scalbeta(i) \gets 0\)}
        \label{alg:ppcr:alpha0beta0}
    \Else
      \State \(\scalbeta(i) \gets \scalgamma(i) / \scalgamma(i-1)\)
        \label{alg:ppcr:beta}
      \State \(\scalalpha(i) \gets \scalgamma(i) /
               (\scaldelta(i) - \scalbeta(i) \scalgamma(i) / \scalalpha(i-1))\)
        \label{alg:ppcr:alpha}
    \EndIf
    \State \(\vecz(i) \gets \vecn(i) + \scalbeta(i) \vecz(i-1)\)
      \label{alg:ppcr:z}
    \State \(\vecq(i) \gets \vecm(i) + \scalbeta(i) \vecq(i-1)\)
      \label{alg:ppcr:q}
    \State \(\vecp(i) \gets \vecu(i) + \scalbeta(i) \vecp(i-1)\)
      \label{alg:ppcr:p}
    \State \(\vecx(i+1) \gets \vecx(i) + \scalalpha(i) \vecp(i)\)
      \label{alg:ppcr:x}
    \State \(\vecu(i+1) \gets \vecu(i) - \scalalpha(i) \vecq(i)\)
      \label{alg:ppcr:u}
    \State \(\vecw(i+1) \gets \vecw(i) - \scalalpha(i) \vecz(i)\)
      \label{alg:ppcr:w}
  \EndFor
\end{algorithmic}
\end{algorithm}

In this section, we review the basic ideas and main properties of the
communication-hiding iterative linear solvers we later consider in the context
of fault tolerance.
We first study the \ppcg method~\cite{Ghysels2014a} in \Sec{solvers:ppcg},
including its foundations in the original \pcg algorithm~\cite{Hestenes1952a,
Saad2003a}.
Subsequently, in \Sec{solvers:other}, we briefly highlight the modifications
that lead to the \ppcr algorithm~\cite{Ghysels2014a} and discuss an alternative
to the \ppcg solver, the \tppcg method~\cite{Eller2016a}.

\subsection{Pipelined \pcg}
\label{sec:solvers:ppcg}

The communication-hiding \ppcg solver~\cite{Ghysels2014a} is a reformulation of
the classical \pcg method~\cite{Hestenes1952a, Saad2003a}.
Let \(\matA \in \R[n][n]\) be an SPD matrix and \(\matMinv \in \R[n][n]\) be an
appropriate SPD preconditioner for \matA, i.e., \(\kappa(\matMinv \matA) <
\kappa(\matA)\), where \(\kappa(\mat{B})\) denotes the condition number of a
matrix \mat{B}.
Furthermore, let \(\vecb \in \R[n]\) and \(\vecx \in \R[n]\) be the
right-hand-side and solution vectors, respectively.
For iteratively solving a given sparse linear system \(\matA \vecx = \vecb\),
the \pcg method, which is listed in \Alg{pcg}, actually solves the
left-preconditioned linear system \(\matMinv \matA \vecx = \matMinv \vecb\)
for accelerated convergence.
After the solver has converged, the iterate \vecx(i) is reasonably close to the
solution vector \vecx.
The search direction vectors \vecp(i) are chosen to be mutually
\matA-orthogonal, i.e., \(\scalprod{\matA \vecp(i)}{\vecp(j)} = 0\) for all
\(i \neq j\).
The residual vector \vecr(i) is defined by the relation \(\vecr(i) = \vecb -
\matA \vecx(i)\).
Additionally, the \pcg solver keeps the preconditioned residual vector
\(\vecu(i) \coloneqq \matMinv \vecr(i)\) and the vector \(\vecs(i) \coloneqq
\matA \vecp(i)\).
For computing the dot products, there are two global synchronization points in
each iteration of \pcg (\Lins{pcg}{delta}{gamma}).
Since the results of those dot products are needed immediately afterwards, both
global communication operations are blocking.

For being able to reorder the \pcg operations such that we have only one global
synchronization point and the possibility to overlap global communication with
computation and local communication, the \ppcg method, which is shown in
\Alg{ppcg}, has to keep the five additional vectors \(\vecw(i) \coloneqq \matA
\vecu(i)\), \(\vecm(i) \coloneqq \matMinv \vecw(i)\), \(\vecn(i) \coloneqq
\matA \vecm(i)\), \(\vecq(i) \coloneqq \matMinv \vecs(i)\), and \(\vecz(i)
\coloneqq \matA \vecq(i)\).
By left-multiplying \matA and \matMinv to relations of the original \pcg
algorithm, Ghysels and Vanroose~\cite{Ghysels2014a} derive new recurrence
relations that allow them to reorder and merge the two dot product computations
to just one global reduction operation (\Lin{ppcg}{gammadelta}).
Furthermore, since the results of the dot products are not needed before
\Lins{ppcg}{alpha0beta0}{beta}, the global reductions can be computed as
\emph{non-blocking} operations and, therefore, can be overlapped with the
application of the preconditioner as well as the SpMV computation in
\Lins{ppcg}{m}{n}.
These two operations typically require only local communication.
However, the \ppcg method has to compute eight (\LINS{ppcg}{z}{w}) instead of
just three vector updates as in the \pcg algorithm.
Although both solvers are mathematically equivalent, we may see different
numerical error propagation in finite precision~\cite{Ghysels2014a}.

\subsection{Other solvers}
\label{sec:solvers:other}

\begin{algorithm}[t!]
\caption{Two-iteration pipelined preconditioned conjugate gradient~(\tppcg)
         method~\cite{Eller2016a}}
\label{alg:2ppcg}
\begin{algorithmic}[1]
  \TriState{\(\vecr(0) \gets \vecb - \matA \vecx(0)\)}
           {\(\vecu(0) \gets \matMinv \vecr(0)\)}
           {\(\vecw(0) \gets \matA \vecu(0)\)}
  \BiState{\(\scalgamma(0) \gets \scalprod{\vecu(0)}{\vecr(0)}\)}
          {\(\scaldelta(0) \gets \scalprod{\vecu(0)}{\vecw(0)}\)}
  \BiState{\(\vecm(0) \gets \matMinv \vecw(0)\)}
          {\(\vecn(0) \gets \matA \vecm(0)\)}
  \BiState{\(\vecc(0) \gets \matMinv \vecn(0)\)}
          {\(\vecd(0) \gets \matA \vecc(0)\)}
  \For{\(i \gets 0, 2, \dots, \text{until convergence}\)}
    \If{\(i = 0\)}
      \TriState{\(\scalzeta(i+1) \gets 1\)}
               {\(\scaleta(i+1) \gets \scalgamma(i) / \scaldelta(i)\)}
               {\(\scaltheta(i+1) \gets 0\)}
    \Else
      \State \(\scaleta(i) \gets \scalgamma(i-1) / \scaldelta(i-1)\)
      \State \(\scalzeta(i) \gets 1 / (1 - \scalgamma(i-1) \scaleta(i) /
               (\scalgamma(i-2) \scalzeta(i-1) \scaleta(i-1)))\)
      \TriState{\(\scalkappa[1] \gets \scalzeta(i)\)}
               {\(\scalkappa[2] \gets -\scalzeta(i) \scaleta(i)\)}
               {\(\scalkappa[3] \gets 1 - \scalzeta(i)\)}
      \State \(\scalgamma(i) \gets
               \scalkappa[1] \scalkappa[1] \scallambda[1] -
               2 \scalkappa[1] \scalkappa[2] \scallambda[7] +
               2 \scalkappa[1] \scalkappa[3] \scallambda[6]\)
      \Statex \(\hspace{\algorithmicindent} \hspace{\algorithmicindent}
                \hphantom{\scalgamma(i) \gets} \mathllap{+} \;
                \scalkappa[2] \scalkappa[2] \scallambda[2] -
                2 \scalkappa[2] \scalkappa[3] \scallambda[3] +
                \scalkappa[3] \scalkappa[3] \scallambda[8]\)
      \State \(\scaldelta(i) \gets
               \scalkappa[1] \scalkappa[1] \scallambda[7] -
               2 \scalkappa[1] \scalkappa[2] \scallambda[2] +
               2 \scalkappa[1] \scalkappa[3] \scallambda[3]\)
      \Statex \(\hspace{\algorithmicindent} \hspace{\algorithmicindent}
                \hphantom{\scaldelta(i) \gets} \mathllap{+} \;
                \scalkappa[2] \scalkappa[2] \scallambda[4] -
                2 \scalkappa[2] \scalkappa[3] \scallambda[5] +
                \scalkappa[3] \scalkappa[3] \scallambda[8]\)
      \State \(\scaleta(i+1) \gets \scalgamma(i) / \scaldelta(i)\)
      \State \(\scalzeta(i+1) \gets 1 / (1 - \scalgamma(i) \scaleta(i+1) /
               (\scalgamma(i-1) \scalzeta(i) \scaleta(i)))\)
      \BiState{\(\scaltheta(i) \gets \scalkappa[3]\)}
              {\(\scaltheta(i+1) \gets 1 - \scalzeta(i+1)\)}
      \State \(\vecx(i) \gets \scalzeta(i) (\vecx(i-1) + \scaleta(i)
               \vecu(i-1)) + \scaltheta(i) \vecx(i-2)\)
      \State \(\vecr(i) \gets \scalzeta(i) (\vecr(i-1) - \scaleta(i)
               \vecw(i-1)) + \scaltheta(i) \vecr(i-2)\)
      \State \(\vecu(i) \gets \scalzeta(i) (\vecu(i-1) - \scaleta(i)
               \vecm(i-1)) + \scaltheta(i) \vecu(i-2)\)
      \State \(\vecw(i) \gets \scalzeta(i) (\vecw(i-1) - \scaleta(i)
               \vecn(i-1)) + \scaltheta(i) \vecw(i-2)\)
      \State \(\vecm(i) \gets \scalzeta(i) (\vecm(i-1) - \scaleta(i)
               \vecc(i-1)) + \scaltheta(i) \vecm(i-2)\)
      \State \(\vecn(i) \gets \scalzeta(i) (\vecn(i-1) - \scaleta(i)
               \vecd(i-1)) + \scaltheta(i) \vecn(i-2)\)
      \State \(\vecc(i) \gets \scalzeta(i) (\vecc(i-1) - \scaleta(i)
               \vecg(i-1)) + \scaltheta(i) \vecc(i-2)\)
      \State \(\vecd(i) \gets \scalzeta(i) (\vecd(i-1) - \scaleta(i)
               \vech(i-1)) + \scaltheta(i) \vecd(i-2)\)
    \EndIf
    \State \(\vecx(i+1) \gets \scalzeta(i+1) (\vecx(i) + \scaleta(i+1)
             \vecu(i)) + \scaltheta(i+1) \vecx(i-1)\)
    \State \(\vecr(i+1) \gets \scalzeta(i+1) (\vecr(i) - \scaleta(i+1)
             \vecw(i)) + \scaltheta(i+1) \vecr(i-1)\)
    \State \(\vecu(i+1) \gets \scalzeta(i+1) (\vecu(i) - \scaleta(i+1)
             \vecm(i)) + \scaltheta(i+1) \vecu(i-1)\)
    \State \(\vecw(i+1) \gets \scalzeta(i+1) (\vecw(i) - \scaleta(i+1)
             \vecn(i)) + \scaltheta(i+1) \vecw(i-1)\)
    \State \(\vecm(i+1) \gets \scalzeta(i+1) (\vecm(i) - \scaleta(i+1)
             \vecc(i)) + \scaltheta(i+1) \vecm(i-1)\)
    \State \(\vecn(i+1) \gets \scalzeta(i+1) (\vecn(i) - \scaleta(i+1)
             \vecd(i)) + \scaltheta(i+1) \vecn(i-1)\)
    \BiState{\(\scallambda[1] \gets \scalprod{\vecu(i+1)}{\vecw(i+1)}\)}
            {\(\scallambda[2] \gets \scalprod{\vecu(i+1)}{\vecw(i)}\)}
      \label{alg:2ppcg:lambda12}
    \BiState{\(\scallambda[3] \gets \scalprod{\vecm(i+1)}{\vecn(i+1)}\)}
            {\(\scallambda[4] \gets \scalprod{\vecm(i+1)}{\vecw(i)}\)}
    \BiState{\(\scallambda[5] \gets \scalprod{\vecu(i)}{\vecw(i)}\)}
            {\(\scallambda[6] \gets \scalprod{\vecu(i+1)}{\vecr(i)}\)}
    \BiState{\(\scallambda[7] \gets \scalprod{\vecu(i)}{\vecr(i)}\)}
            {\(\scallambda[8] \gets \scalprod{\vecu(i+1)}{\vecu(i+1)}\)}
      \label{alg:2ppcg:lambda78}
    \BiState{\(\scalgamma(i+1) \gets \scalprod{\vecu(i+1)}{\vecr(i+1)}\)}
            {\(\scaldelta(i+1) \gets \scallambda[1]\)}
      \label{alg:2ppcg:gammadelta}
    \BiState{\(\vecc(i+1) \gets \matMinv \vecn(i+1)\)}
            {\(\vecd(i+1) \gets \matA \vecc(i+1)\)}
      \label{alg:2ppcg:cd}
    \BiState{\(\vecg(i+1) \gets \matMinv \vecd(i+1)\)}
            {\(\vech(i+1) \gets \matA \vecg(i+1)\)}
      \label{alg:2ppcg:gh}
  \EndFor
\end{algorithmic}
\end{algorithm}

When the \matM inner product that is used for the dot products in \ppcg is
replaced by the \matA inner product (\(\matMinv \matA\) is also self-adjoint
with respect to this inner product), we obtain the \ppcr method, which is
listed in \Alg{ppcr}, as a variation of the \ppcg solver~\cite{Ghysels2014a}.
Since the preconditioner has then to be applied before the dot products, the
merged global reduction operation (\Lin{ppcr}{gammadelta}) can only be
overlapped with the SpMV computation (\Lin{ppcr}{n}).
Moreover, there is no dependence on \vecr(i) and \vecs(i) anymore and, hence,
those vectors do not need to be updated in every iteration.
In this case, the convergence criterion can be based on the preconditioned
residual \vecu(i) instead of the residual \vecr(i).

Eller and Gropp~\cite{Eller2016a} suggest an alternative pipelined \pcg
algorithm based on a \pcg variant with two three-term instead of three two-term
recurrences.
The resulting \tppcg method, which is shown in \Alg{2ppcg}, computes two \pcg
iterations at once.
In addition to the vectors \vecx(i), \vecr(i), \vecu(i), \vecw(i), \vecm(i),
and \vecn(i) already known from \ppcg, it keeps and updates the vectors
\(\vecc(i) \coloneqq \matMinv \vecn(i)\), \(\vecd(i) \coloneqq \matA
\vecc(i)\), \(\vecg(i) \coloneqq \matMinv \vecd(i)\), and \(\vech(i) \coloneqq
\matA \vecg(i)\).
All except the latter two vectors have to be stored for both consecutive
iterations that are computed together.
The \tppcg solver merges multiple dot products into one global reduction
operation (\LINS{2ppcg}{lambda12}{gammadelta}).
This global communication operation is overlapped with the computation of two
preconditioner applications and two SpMV computations (\LINS{2ppcg}{cd}{gh}).

\section{Data redundancy}
\label{sec:redundancy}

In order to attain algorithm-specific resilience against node failures for the
considered communication-hiding \pcg and \pcr solvers, we have to take two
separate aspects into account.
On the one hand, we need to have recovery procedures that reconstruct the full
state of the solver after node failures occurred.
We discuss the recovery process after unexpected node failures in
\Sec{recovery}.
However, for those recovery procedures to work properly, we need to have some
guaranteed data redundancy.
Hence, on the other hand, we need to exploit the specific properties of the
solvers to achieve the required minimum level of data redundancy as
cost-effective as possible.
The outlined strategy for data redundancy we now apply to communication-hiding
solvers has originally been proposed for the classical \pcg
method~\cite{Chen2011a, Pachajoa2018a, Pachajoa2019a}.

All of the solvers reviewed in \Sec{solvers} compute at least one SpMV per
iteration.
We particularly consider the computation of \(\matA \vecp(i)\) in \pcg
(\Lin{pcg}{s}), \(\matA \vecm(i)\) in \ppcg and \ppcr (\Lin{ppcg}{n} and
\Lin{ppcr}{n}), and \(\matA \vecc(i+1)\) in \tppcg (\Lin{2ppcg}{cd}).
During the SpMV computation, vector elements from other nodes---for many sparse
matrices especially from \emph{neighbor} nodes, i.e., usually \emph{local}
communication is sufficient---are required on node \(j\), \(j \in \{1, 2,
\dots, \nn\}\).
In the non-resilient standard solvers, all but the elements of the own block
(\vecp[j](i) in \pcg, \vecm[j](i) in \ppcg and \ppcr, and \vecc[j](i+1) in
\tppcg) can be dropped on node \(j\) after the product has been computed.
For recovering the full solver state, we need to have entire copies of the
vectors involved in SpMV from the latest \emph{two} solver iterations, i.e.,
\vecp(i-1) and \vecp(i) for \pcg, \vecm(i-1) and \vecm(i) for \ppcg and \ppcr,
or \vecc(i) and \vecc(i+1) for \tppcg (cf.\ \Sec{recovery}).
Hence, the vector blocks \vecp[\rn](i-1) and \vecp[\rn](i) (for \pcg),
\vecm[\rn](i-1) and \vecm[\rn](i) (for \ppcg and \ppcr), or \vecc[\rn](i) and
\vecc[\rn](i+1) (for \tppcg) of the failed node \rn must be available as well at
the beginning of the recovery process.
Thus, we have to make sure to keep enough redundant copies of each vector
element on other nodes than the owner after the SpMV computation in each solver
iteration, instead of dropping all of them as in the non-resilient standard
variant.
For the \tppcg solver, we store the redundant copies of \vecc[\rn](i) together
with those of \vecc[\rn](i+1) during computing \(\matA \vecc(i+1)\) (since this
solver computes two iterations at a time, cf.\ \Sec{solvers:other}).

\begin{algorithm}[t]
\caption{Node failure recovery (on replacement node~\rn) for the \pcg
         method~\cite[Alg.~4]{Pachajoa2018a}}
\label{alg:esrpcg}
\begin{algorithmic}[1]
\State Gather \vecr[\compl{\rn}](\nfi) and \vecx[\compl{\rn}](\nfi)
\State Retrieve static data \matA[\rn][\all], \matP[\rn][\all], and \vecb[\rn]
\State \label{alg:esrpcg:redundant}Retrieve redundant copies of
       \scalbeta(\nfi-1), \vecp[\rn](\nfi-1), and \vecp[\rn](\nfi)
\State Compute \(\vecu[\rn](\nfi) \gets
       \vecp[\rn](\nfi) - \scalbeta(\nfi-1) \vecp[\rn](\nfi-1)\)
\State Compute \(\vectilde[\rn](\nfi){u} \gets \vecu[\rn](\nfi) -
       \matP[\rn][\compl{\rn}] \vecr[\compl{\rn}](\nfi)\)
\State Solve \(\matP[\rn][\rn] \vecr[\rn](\nfi) = \vectilde[\rn](\nfi){u}\) for
       \vecr[\rn](\nfi)
       \label{alg:esrpcg:r}
\State Compute \(\vectilde[\rn](\nfi){b} \gets \vecb[\rn] - \vecr[\rn](\nfi) -
       \matA[\rn][\compl{\rn}] \vecx[\compl{\rn}](\nfi)\)
\State Solve \(\matA[\rn][\rn] \vecx[\rn](\nfi) = \vectilde[\rn](\nfi){b}\) for
       \vecx[\rn](\nfi)
       \label{alg:esrpcg:x}
\State Continue in \Lin{pcg}{s} at iteration~\nfi
\end{algorithmic}
\end{algorithm}

\begin{algorithm*}[t]
\caption{Node failure recovery (on replacement node~\rn) for the \ppcg method}
\label{alg:esrppcg}
\begin{algorithmic}[1]
\State Gather \vecr[\compl{\rn}](\nfi-1), \vecr[\compl{\rn}](\nfi),
       \vecu[\compl{\rn}](\nfi-1), \vecu[\compl{\rn}](\nfi),
       \vecw[\compl{\rn}](\nfi-1), \vecw[\compl{\rn}](\nfi),
       \vecx[\compl{\rn}](\nfi-1), and \vecx[\compl{\rn}](\nfi)
       \label{alg:esrppcg:gather}
\State Retrieve static data \matA[\rn][\all], \matP[\rn][\all], and \vecb[\rn]
\State Retrieve redundant copies of \scalalpha(\nfi-1), \scalgamma(\nfi-1),
       \scalgamma(\nfi), \scaldelta(\nfi), \vecm[\rn](\nfi-1), and
       \vecm[\rn](\nfi)
       \label{alg:esrppcg:redundant}
\State Compute \(\vecconcat{\vectilde[\rn](\nfi-1){m}}{\vectilde[\rn](\nfi){m}}
       \gets \vecconcat{\vecm[\rn](\nfi-1)}{\vecm[\rn](\nfi)} -
       \matP[\rn][\compl{\rn}]
       \vecconcat{\vecw[\compl{\rn}](\nfi-1)}{\vecw[\compl{\rn}](\nfi)}\)
       \label{alg:esrppcg:mtilde}
\State Solve \(\matP[\rn][\rn] \vecconcat{\vecw[\rn](\nfi-1)}{\vecw[\rn](\nfi)}
       = \vecconcat{\vectilde[\rn](\nfi-1){m}}{\vectilde[\rn](\nfi){m}}\) for
       \vecconcat{\vecw[\rn](\nfi-1)}{\vecw[\rn](\nfi)}
\State Compute \(\vecconcat{\vectilde[\rn](\nfi-1){w}}{\vectilde[\rn](\nfi){w}}
       \gets \vecconcat{\vecw[\rn](\nfi-1)}{\vecw[\rn](\nfi)} -
       \matA[\rn][\compl{\rn}]
       \vecconcat{\vecu[\compl{\rn}](\nfi-1)}{\vecu[\compl{\rn}](\nfi)}\)
\State Solve \(\matA[\rn][\rn] \vecconcat{\vecu[\rn](\nfi-1)}{\vecu[\rn](\nfi)}
       = \vecconcat{\vectilde[\rn](\nfi-1){w}}{\vectilde[\rn](\nfi){w}}\) for
       \vecconcat{\vecu[\rn](\nfi-1)}{\vecu[\rn](\nfi)}
\State Compute \(\vecconcat{\vectilde[\rn](\nfi-1){u}}{\vectilde[\rn](\nfi){u}}
       \gets \vecconcat{\vecu[\rn](\nfi-1)}{\vecu[\rn](\nfi)} -
       \matP[\rn][\compl{\rn}]
       \vecconcat{\vecr[\compl{\rn}](\nfi-1)}{\vecr[\compl{\rn}](\nfi)}\)
\State Solve \(\matP[\rn][\rn] \vecconcat{\vecr[\rn](\nfi-1)}{\vecr[\rn](\nfi)}
       = \vecconcat{\vectilde[\rn](\nfi-1){u}}{\vectilde[\rn](\nfi){u}}\) for
       \vecconcat{\vecr[\rn](\nfi-1)}{\vecr[\rn](\nfi)}
\State Compute \(\vecconcat{\vectilde[\rn](\nfi-1){b}}{\vectilde[\rn](\nfi){b}}
       \gets \vecconcat{\vecb[\rn]}{\vecb[\rn]} -
       \vecconcat{\vecr[\rn](\nfi-1)}{\vecr[\rn](\nfi)} -
       \matA[\rn][\compl{\rn}]
       \vecconcat{\vecx[\compl{\rn}](\nfi-1)}{\vecx[\compl{\rn}](\nfi)}\)
\State Solve \(\matA[\rn][\rn] \vecconcat{\vecx[\rn](\nfi-1)}{\vecx[\rn](\nfi)}
       = \vecconcat{\vectilde[\rn](\nfi-1){b}}{\vectilde[\rn](\nfi){b}}\) for
       \vecconcat{\vecx[\rn](\nfi-1)}{\vecx[\rn](\nfi)}
       \label{alg:esrppcg:x}
\State Compute \(\vecz[\rn](\nfi-1) \gets \nicefrac{1}{\scalalpha(\nfi-1)}
       (\vecw[\rn](\nfi-1) - \vecw[\rn](\nfi))\)
       \label{alg:esrppcg:z}
\State Compute \(\vecq[\rn](\nfi-1) \gets \nicefrac{1}{\scalalpha(\nfi-1)}
       (\vecu[\rn](\nfi-1) - \vecu[\rn](\nfi))\)
\State Compute \(\vecs[\rn](\nfi-1) \gets \nicefrac{1}{\scalalpha(\nfi-1)}
       (\vecr[\rn](\nfi-1) - \vecr[\rn](\nfi))\)
\State Compute \(\vecp[\rn](\nfi-1) \gets \nicefrac{1}{\scalalpha(\nfi-1)}
       (\vecx[\rn](\nfi) - \vecx[\rn](\nfi-1))\)
       \label{alg:esrppcg:p}
\State Continue in \Lin{ppcg}{n} at iteration~\nfi
\end{algorithmic}
\end{algorithm*}

However, depending on the sparsity pattern of \matA, not all vector elements
are necessarily sent to other nodes during the SpMV computation.
For this reason, we employ a strategy that guarantees enough redundant copies
of each vector element after the SpMV operation while preferring local over
global communication~\cite{Chen2011a, Pachajoa2018a, Pachajoa2019a}.
For many practical scenarios, it is adequate to support only one node failure
at a time.
Consequently, the recovery process has to be finished before another node
failure may occur.
In this case, we have to ensure one \emph{redundant} copy of each vector
element, i.e., one copy \emph{additional} to the copy of the owner.
We present a generalized redundancy strategy that is capable of supporting
\(1 \leq \nredu < \nn\) \emph{simultaneous} node failures by keeping \nredu
redundant copies of each vector element on \nredu nodes different from its
owner~\cite{Pachajoa2019a}.

Let \(S_j\) be the set of all elements of \vecp[j](i) (for \pcg), \vecm[j](i)
(for \ppcg and \ppcr), or \vecc[j](i+1) (for \tppcg), and let \(S_{jk}\) denote
the set of all elements of \(S_j\) sent to node \(k\) during the computation of
\(\matA \vecp(i)\) (for \pcg), \(\matA \vecm(i)\) (for \ppcg and \ppcr), or
\(\matA \vecc(i+1)\) (for \tppcg).
Furthermore, let \((S_j, m_j)\) denote a multiset with the multiplicity
\begin{equation}
\label{eqn:redundancy:mj}
\begin{aligned}
m_j \colon S_j \to~&\N_0 \\
s \mapsto~&\text{number of nodes}~s~\text{is sent to} \\[-0.7ex]
          &\text{during the SpMV computation,}
\end{aligned}
\end{equation}
let the function \(d_{jk}\) be defined as
\begin{equation}
\label{eqn:redundancy:djk}
d_{jk} \coloneqq \left\{
\begin{array}{ll}
\left( j + \left\lceil \frac{k}{2} \right\rceil \right) \bmod \nn, &
  \text{if}~k~\text{odd} \\
\left( j - \frac{k}{2} \right) \bmod \nn, & \text{if}~k~\text{even,}
\end{array}
\right.
\end{equation}
and let \(g_j(s)\) be the number of sets \(S_{j d_{jk}}\) with \(s \in
S_{j d_{jk}}\) for all \(k \in \{1, 2, \dots, \nredu\}\).
Then, the necessary data redundancy for tolerating up to \nredu simultaneous
node failures is guaranteed by sending the elements of the set
\begin{equation}
\label{eqn:redundancy:Rjk}
\begin{aligned}
R_{jk} \coloneqq \{ s \in S_j \; | \; &s \notin S_{j d_{jk}} \; \land \\
                                      &m_j(s) - g_j(s) \leq \nredu - k \}
\end{aligned}
\end{equation}
to node \(d_{jk}\) for all \(j \in \{1, 2, \dots, \nn\}\) and \(k \in \{1, 2,
\dots, \nredu\}\) during the SpMV computation~\cite{Pachajoa2019a}.
\(R_{jk}\) always is of minimal size and it holds that \(|R_{j1}| \geq |R_{j2}|
\geq \dots \geq |R_{j \nredu}|\).
Note that the elements of \(R_{jk}\) are sent to node \(d_{jk}\) together with
the elements of \(S_{j d_{jk}}\), which have to be sent anyway according to the
sparsity pattern of \matA.
Hence, in many cases, no extra latency cost is incurred for establishing new
connections.

According to our strategy, the nodes selected for receiving the redundant
vector element copies always are close neighbor nodes
(cf.\ \Eqn{redundancy:djk}).
Therefore, the communication overhead compared to the non-resilient standard
SpMV computation only consists of \emph{local} communication and, thus, is
perfectly appropriate for overlapping the global reduction operations in the
communication-hiding \pcg and \pcr solvers.
It can theoretically be shown~\cite{Pachajoa2019a} that the communication
overhead for keeping \nredu redundant copies of all elements of the vector
involved in the SpMV is bounded between 0 and \(\nredu (\latency_{\max} +
\lceil \nicefrac{n}{\nn} \rceil \invbw)\), where \(\latency_{\max}\) is the
maximum latency for establishing a new connection and \invbw is the
communication cost per vector element.
However, for large-scale parallel computers, it is plausible that---even in the
worst case of maximum local communication overhead for the data redundancy
during the SpMV---the global reduction operation is more expensive than the
preconditioner application and SpMV computation together, which are overlapping
the global reduction in the \ppcg, \ppcr, and \tppcg solvers.

\section{Recovery from node failures}
\label{sec:recovery}

\begin{algorithm*}[t]
\caption{Node failure recovery (on replacement node~\rn) for the \ppcr method}
\label{alg:esrppcr}
\begin{algorithmic}[1]
\State Gather \vecu[\compl{\rn}](\nfi-1), \vecu[\compl{\rn}](\nfi),
       \vecw[\compl{\rn}](\nfi-1), \vecw[\compl{\rn}](\nfi),
       \vecx[\compl{\rn}](\nfi-1), and \vecx[\compl{\rn}](\nfi)
\State Retrieve static data \matA[\rn][\all], \matP[\rn][\all], and \vecb[\rn]
\State Retrieve redundant copies of \scalalpha(\nfi-1), \scalgamma(\nfi-1),
       \scalgamma(\nfi), \scaldelta(\nfi), \vecm[\rn](\nfi-1), and
       \vecm[\rn](\nfi)
\State Compute \(\vecconcat{\vectilde[\rn](\nfi-1){m}}{\vectilde[\rn](\nfi){m}}
       \gets \vecconcat{\vecm[\rn](\nfi-1)}{\vecm[\rn](\nfi)} -
       \matP[\rn][\compl{\rn}]
       \vecconcat{\vecw[\compl{\rn}](\nfi-1)}{\vecw[\compl{\rn}](\nfi)}\)
\State Solve \(\matP[\rn][\rn] \vecconcat{\vecw[\rn](\nfi-1)}{\vecw[\rn](\nfi)}
       = \vecconcat{\vectilde[\rn](\nfi-1){m}}{\vectilde[\rn](\nfi){m}}\) for
       \vecconcat{\vecw[\rn](\nfi-1)}{\vecw[\rn](\nfi)}
\State Compute \(\vecconcat{\vectilde[\rn](\nfi-1){w}}{\vectilde[\rn](\nfi){w}}
       \gets \vecconcat{\vecw[\rn](\nfi-1)}{\vecw[\rn](\nfi)} -
       \matA[\rn][\compl{\rn}]
       \vecconcat{\vecu[\compl{\rn}](\nfi-1)}{\vecu[\compl{\rn}](\nfi)}\)
\State Solve \(\matA[\rn][\rn] \vecconcat{\vecu[\rn](\nfi-1)}{\vecu[\rn](\nfi)}
       = \vecconcat{\vectilde[\rn](\nfi-1){w}}{\vectilde[\rn](\nfi){w}}\) for
       \vecconcat{\vecu[\rn](\nfi-1)}{\vecu[\rn](\nfi)}
\State Compute \(\vecconcat{\vectilde[\rn](\nfi-1){b}}{\vectilde[\rn](\nfi){b}}
       \gets \matP[\rn][\all] \vecconcat{\vecb[\rn]}{\vecb[\rn]} -
       \vecconcat{\vecu[\rn](\nfi-1)}{\vecu[\rn](\nfi)} -
       \matP[\rn][\all] (\matA[\all][\compl{\rn}]
       \vecconcat{\vecx[\compl{\rn}](\nfi-1)}{\vecx[\compl{\rn}](\nfi)})\)
       \label{alg:esrppcr:btilde}
\State Solve \((\matP[\rn][\all] \matA[\all][\rn])
       \vecconcat{\vecx[\rn](\nfi-1)}{\vecx[\rn](\nfi)}
       = \vecconcat{\vectilde[\rn](\nfi-1){b}}{\vectilde[\rn](\nfi){b}}\) for
       \vecconcat{\vecx[\rn](\nfi-1)}{\vecx[\rn](\nfi)}
       \label{alg:esrppcr:x}
\State Compute \(\vecz[\rn](\nfi-1) \gets \nicefrac{1}{\scalalpha(\nfi-1)}
       (\vecw[\rn](\nfi-1) - \vecw[\rn](\nfi))\)
\State Compute \(\vecq[\rn](\nfi-1) \gets \nicefrac{1}{\scalalpha(\nfi-1)}
       (\vecu[\rn](\nfi-1) - \vecu[\rn](\nfi))\)
\State Compute \(\vecp[\rn](\nfi-1) \gets \nicefrac{1}{\scalalpha(\nfi-1)}
       (\vecx[\rn](\nfi) - \vecx[\rn](\nfi-1))\)
\State Continue in \Lin{ppcr}{n} at iteration~\nfi
\end{algorithmic}
\end{algorithm*}

After discussing how to ensure sufficient data redundancy in \Sec{redundancy},
we now focus on the recovery process after actual node failures.
Our goal is to reconstruct the full state of the communication-hiding iterative
solver such that it is able to continue as if no node failure occurred.
First, in \Sec{recovery:ppcg}, we illustrate the main principles and derive the
recovery procedure for the \ppcg method.
Subsequently, in \Sec{recovery:other}, we highlight the differences for the
\ppcr and \tppcg solvers and show their recovery procedures.

\subsection{Recovery process for pipelined \pcg}
\label{sec:recovery:ppcg}

We can distinguish between \emph{static} and \emph{dynamic} iterative solver
data.
Static data is defined as input data that does not change during the execution
of the iterative solver.
Analogous to Chen~\cite{Chen2011a} and Agullo et al.~\cite{Agullo2016a}, we
assume that static data can always be retrieved from reliable external storage
like a checkpoint taken prior to entering the iterative solver.
For the \ppcg method (as well as the \pcg, \ppcr, and \tppcg solvers), the
system matrix \matA, the preconditioner \matM, and the right-hand-side vector
\vecb are considered to be static data.
On the other hand, dynamic solver data is continuously modified by the
iterative solver and can be either equal on all nodes or unique to each node.
For our \pcg and \pcr variants, scalars that are results of global reduction
operation are equal on all \nn nodes.
Those scalars can hence easily be retrieved from any of the surviving nodes
after a node failure.
In contrast, dynamic vector data is unique to each node since each vector is
distributed among all \nn nodes.
Therefore, parts of those vectors are lost in case of a node failure and need
to be reconstructed on the replacement node.

We first focus on the recovery process for the case of a single node failure.
The generalization to multiple simultaneous node failures will later be
straightforward.
For simplifying the notation, we define \(\matP \coloneqq \matMinv\).
Furthermore, we assume that \matP (not \matM) is given as input data of the
iterative solver.
For the classical \pcg method, Chen~\cite{Chen2011a} derives a procedure for
recovery after a node failure.
We show a variant of this recovery procedure by Pachajoa et
al.~\cite{Pachajoa2018a} in \Alg{esrpcg}.
In \Lin{esrpcg}{redundant}, the redundant copies of the lost elements of the
latest two search direction vectors, \vecp[\rn](\nfi-1) and \vecp[\rn](\nfi),
are retrieved from the backup nodes according to the redundancy strategy
described in \Sec{redundancy}.
In \Lins{esrpcg}{r}{x}, two local linear systems are solved locally on the
replacement node \rn.
Note that these local systems are typically very small compared to the given
linear system \(\matA \vecx = \vecb\).

\begin{algorithm*}[t]
\caption{Node failure recovery (on replacement node~\rn) for the \tppcg method}
\label{alg:esr2ppcg}
\begin{algorithmic}[1]
\State Gather \vecc[\compl{\rn}](\nfi), \vecc[\compl{\rn}](\nfi+1),
       \vecm[\compl{\rn}](\nfi), \vecm[\compl{\rn}](\nfi+1),
       \vecn[\compl{\rn}](\nfi), \vecn[\compl{\rn}](\nfi+1),
       \vecr[\compl{\rn}](\nfi), \vecr[\compl{\rn}](\nfi+1),
       \vecu[\compl{\rn}](\nfi), \vecu[\compl{\rn}](\nfi+1),
       \vecw[\compl{\rn}](\nfi), \vecw[\compl{\rn}](\nfi+1),
       \vecx[\compl{\rn}](\nfi), and \vecx[\compl{\rn}](\nfi+1)
\State Retrieve static data \matA[\rn][\all], \matP[\rn][\all], and \vecb[\rn]
\State Retrieve redundant copies of \scalgamma(\nfi), \scalgamma(\nfi+1),
       \scalzeta(\nfi+1), \scaleta(\nfi+1), \scallambda[1], \scallambda[2],
       \scallambda[3], \scallambda[4], \scallambda[5], \scallambda[6],
       \scallambda[7], \scallambda[8],
       \vecc[\rn](\nfi), and \vecc[\rn](\nfi+1)
\State Assign \(\scaldelta(i+1) \gets \scallambda[1]\)
\State Compute \(\vecconcat{\vecd[\rn](\nfi)}{\vecd[\rn](\nfi+1)}
       \gets \matA[\rn][\all] \vecconcat{\vecc(\nfi)}{\vecc(\nfi+1)}\)
\State Compute \(\vecconcat{\vectilde[\rn](\nfi){c}}{\vectilde[\rn](\nfi+1){c}}
       \gets \vecconcat{\vecc[\rn](\nfi)}{\vecc[\rn](\nfi+1)} -
       \matP[\rn][\compl{\rn}]
       \vecconcat{\vecn[\compl{\rn}](\nfi)}{\vecn[\compl{\rn}](\nfi+1)}\)
\State Solve \(\matP[\rn][\rn] \vecconcat{\vecn[\rn](\nfi)}{\vecn[\rn](\nfi+1)}
       = \vecconcat{\vectilde[\rn](\nfi){c}}{\vectilde[\rn](\nfi+1){c}}\) for
       \vecconcat{\vecn[\rn](\nfi)}{\vecn[\rn](\nfi+1)}
\State Compute \(\vecconcat{\vectilde[\rn](\nfi){n}}{\vectilde[\rn](\nfi+1){n}}
       \gets \vecconcat{\vecn[\rn](\nfi)}{\vecn[\rn](\nfi+1)} -
       \matA[\rn][\compl{\rn}]
       \vecconcat{\vecm[\compl{\rn}](\nfi)}{\vecm[\compl{\rn}](\nfi+1)}\)
\State Solve \(\matA[\rn][\rn] \vecconcat{\vecm[\rn](\nfi)}{\vecm[\rn](\nfi+1)}
       = \vecconcat{\vectilde[\rn](\nfi){n}}{\vectilde[\rn](\nfi+1){n}}\) for
       \vecconcat{\vecm[\rn](\nfi)}{\vecm[\rn](\nfi+1)}
\State Compute \(\vecconcat{\vectilde[\rn](\nfi){m}}{\vectilde[\rn](\nfi+1){m}}
       \gets \vecconcat{\vecm[\rn](\nfi)}{\vecm[\rn](\nfi+1)} -
       \matP[\rn][\compl{\rn}]
       \vecconcat{\vecw[\compl{\rn}](\nfi)}{\vecw[\compl{\rn}](\nfi+1)}\)
\State Solve \(\matP[\rn][\rn] \vecconcat{\vecw[\rn](\nfi)}{\vecw[\rn](\nfi+1)}
       = \vecconcat{\vectilde[\rn](\nfi){m}}{\vectilde[\rn](\nfi+1){m}}\) for
       \vecconcat{\vecw[\rn](\nfi)}{\vecw[\rn](\nfi+1)}
\State Compute \(\vecconcat{\vectilde[\rn](\nfi){w}}{\vectilde[\rn](\nfi+1){w}}
       \gets \vecconcat{\vecw[\rn](\nfi)}{\vecw[\rn](\nfi+1)} -
       \matA[\rn][\compl{\rn}]
       \vecconcat{\vecu[\compl{\rn}](\nfi)}{\vecu[\compl{\rn}](\nfi+1)}\)
\State Solve \(\matA[\rn][\rn] \vecconcat{\vecu[\rn](\nfi)}{\vecu[\rn](\nfi+1)}
       = \vecconcat{\vectilde[\rn](\nfi){w}}{\vectilde[\rn](\nfi+1){w}}\) for
       \vecconcat{\vecu[\rn](\nfi)}{\vecu[\rn](\nfi+1)}
\State Compute \(\vecconcat{\vectilde[\rn](\nfi){u}}{\vectilde[\rn](\nfi+1){u}}
       \gets \vecconcat{\vecu[\rn](\nfi)}{\vecu[\rn](\nfi+1)} -
       \matP[\rn][\compl{\rn}]
       \vecconcat{\vecr[\compl{\rn}](\nfi)}{\vecr[\compl{\rn}](\nfi+1)}\)
\State Solve \(\matP[\rn][\rn] \vecconcat{\vecr[\rn](\nfi)}{\vecr[\rn](\nfi+1)}
       = \vecconcat{\vectilde[\rn](\nfi){u}}{\vectilde[\rn](\nfi+1){u}}\) for
       \vecconcat{\vecr[\rn](\nfi)}{\vecr[\rn](\nfi+1)}
\State Compute \(\vecconcat{\vectilde[\rn](\nfi){b}}{\vectilde[\rn](\nfi+1){b}}
       \gets \vecconcat{\vecb[\rn]}{\vecb[\rn]} -
       \vecconcat{\vecr[\rn](\nfi)}{\vecr[\rn](\nfi+1)} -
       \matA[\rn][\compl{\rn}]
       \vecconcat{\vecx[\compl{\rn}](\nfi)}{\vecx[\compl{\rn}](\nfi+1)}\)
\State Solve \(\matA[\rn][\rn] \vecconcat{\vecx[\rn](\nfi)}{\vecx[\rn](\nfi+1)}
       = \vecconcat{\vectilde[\rn](\nfi){b}}{\vectilde[\rn](\nfi+1){b}}\) for
       \vecconcat{\vecx[\rn](\nfi)}{\vecx[\rn](\nfi+1)}
\State Continue in \Lin{2ppcg}{gh} at iteration~\nfi
\end{algorithmic}
\end{algorithm*}

We now derive a recovery procedure that reconstructs the full state of the
communication-hiding \ppcg solver.
\begin{lemma}
\label{lem:locsys}
Let \(\mat{B} \vec{y} = \vec{v}\), where \(\mat{B} \in \R[n][n]\) is
an SPD matrix, \(\vec{y} \in \R[n]\), and \(\vec{v} \in \R[n]\).
Then, the lost elements \vec[\rn]{y} of the vector \vec{y} can be
reconstructed after a node failure by solving the linear system
\[\mat[\rn][\rn]{B} \vec[\rn]{y} = \vec[\rn]{v} -
  \mat[\rn][\compl{\rn}]{B} \vec[\compl{\rn}]{y},\]
where \mat[\rn][\rn]{B} has full rank.
\end{lemma}
\begin{proof}
Due to \(\mat{B} \vec{y} = \vec{v}\), it holds that
\(\mat[\rn][\all]{B} \vec{y} = \vec[\rn]{v}\).
By reordering the columns of \mat[\rn][\all]{B} and the rows of \vec{y}, it
follows that
\begin{alignat*}{2}
& &
  \begin{pmatrix} \mat[\rn][\rn]{B} &
                  \mat[\rn][\compl{\rn}]{B} \end{pmatrix}
  \begin{pmatrix} \vec[\rn]{y} \\
                  \vec[\compl{\rn}]{y} \end{pmatrix}
  &= \vec[\rn]{v} \\
&\iff & \mat[\rn][\rn]{B} \vec[\rn]{y}
  &= \vec[\rn]{v} -
     \mat[\rn][\compl{\rn}]{B} \vec[\compl{\rn}]{y}.
\end{alignat*}
Since \mat{B} is an SPD matrix, the square diagonal block \mat[\rn][\rn]{B} is
non-singular and, thus, the linear system has a unique solution.
\end{proof}
After \vec[\compl{\rn}]{y} has been gathered from the other nodes, the
linear system \(\mat[\rn][\rn]{B} \vec[\rn]{y} = \vec[\rn]{v} -
\mat[\rn][\compl{\rn}]{B} \vec[\compl{\rn}]{y}\) can be solved locally on
the replacement node \rn.
Similar to the local linear systems in the recovery procedure for \pcg, this
linear system typically is very small compared to \(\mat{B} \vec{y} =
\vec{v}\).

The recovery procedure for the \ppcg solver is listed in \Alg{esrppcg}.
In \Lin{esrppcg}{redundant}, the redundant copies of \vecm[\rn](\nfi-1) and
\vecm[\rn](\nfi) (cf.\ \Sec{redundancy}) are retrieved.
Then, in \LINS{esrppcg}{mtilde}{x}, eight local linear systems are solved
(possibly pairwise).
Those linear systems can be derived by applying \Lem{locsys} to the
vector-defining equations \(\matP \vecw(i) = \vecm(i)\), \(\matA \vecu(i) =
\vecw(i)\), and \(\matP \vecr(i) = \vecu(i)\) (cf.\ \Sec{solvers:ppcg}) as well
as the residual relation \(\matA \vecx(i) = \vecb - \vecr(i)\).
Afterwards, in \LINS{esrppcg}{z}{p}, the lost elements of the remaining
vectors are locally computed based on the results of the previously solved
linear systems.
Those equations can be obtained by rearranging the \ppcg recurrence relations
in \LINS{ppcg}{x}{w}.
The gather operations in \Lin{esrppcg}{gather} may be non-blocking in order to
overlap them with computation.

If the node failure occurs during iteration \nfi, \Alg{esrppcg} reconstructs
the state at iteration \nfi if (and only if) both the global reduction
operation (\Lin{ppcg}{gammadelta}) and the SpMV computation (\Lin{ppcg}{n})
are already finished prior to the node failure.
Else, it recovers the state at iteration \(\nfi - 1\).
In case of multiple node failures, \(\nfail \leq \nredu\)
(cf.\ \Sec{redundancy}) simultaneous node failures occur.
Let \(\rn_1, \rn_2, \dots, \rn_{\nfail}\) be the nodes that fail.
Then, \Alg{esrppcg} is also appropriate for recovering from multiple
simultaneous node failures if we define the subscript \rn in \Alg{esrppcg} to
denote the union of the indices of the rows owned by nodes \(\rn_1, \rn_2,
\dots, \rn_{\nfail}\) (cf.\ \Sec{introduction:terminology}).
Some of the recovery steps of \Alg{esrppcg} can be performed locally on each of
the replacement nodes.
However, for computing the SpMVs and solving the linear systems in
\LINS{esrppcg}{mtilde}{x}, additional communication between the \nfail
replacement nodes is required.

\subsection{Recovery processes for other solvers}
\label{sec:recovery:other}

\Algs{esrppcr}{esr2ppcg} show the recovery processes for the \ppcr and \tppcg
methods, respectively.
In contrast to \Alg{esrppcg}, the vector elements \vecr[\rn](i-1),
\vecr[\rn](i), and \vecs[\rn](i-1) are not computed in \Alg{esrppcr} since
the corresponding vectors are not available in \ppcr.
Hence, the residual \vecr(i) has to be replaced by the preconditioned residual
\vecu(i) for restoring the lost elements \vecx[\rn](i) of the iterate.
This can be achieved by replacing \(\matA \vecx(i)\) with \(\matA[\all][\rn]
\vecx[\rn](i) + \matA[\all][\compl{\rn}] \vecx[\compl{\rn}](i)\) in the
residual relation \(\matA \vecx(i) = \vecb - \vecr(i)\) and then by
left-multiplying the residual relation with \matP[\rn][\all], which leads to
\Lins{esrppcr}{btilde}{x}.
\Alg{esr2ppcg} does not use any rearranged recurrence relations.
Instead, it solves in total twelve local linear systems derived with
\Lem{locsys}.

\section{Experiments}
\label{sec:experiments}

In this section, we describe our implementation and experimental evaluation of
the \ppcg method (cf.\ \Sec{solvers:ppcg} and \Alg{ppcg}) and our novel
algorithm for protecting the \ppcg solver against node failures
(cf.\ \Secs{redundancy}{recovery:ppcg} as well as \Alg{esrppcg}).
We show experimental results of our resilient \ppcg solver in comparison to
the non-resilient standard \ppcg solver measured on a small high-performance
computer cluster.

\subsection{Test data}
\label{sec:experiments:data}

The matrices used in our experiments were taken from the SuiteSparse Matrix
Collection~\cite{Davis2011a}.
\Tab{experiments:matrices} summarizes their most important properties.
Ghysels and Vanroose~\cite{Ghysels2014a} point out that the \ppcg algorithm is
best suited for matrices where the SpMV only requires local communication
between neighboring nodes (i.e., banded matrices), since overlapping the global
dot product with the SpMV computation will yield the best results in these
cases.
Among our test matrices, M1--M4, M8, and M9 fulfill this criterion.

\begin{figure}[t]
\begin{center}
\includegraphics[width=\linewidth]{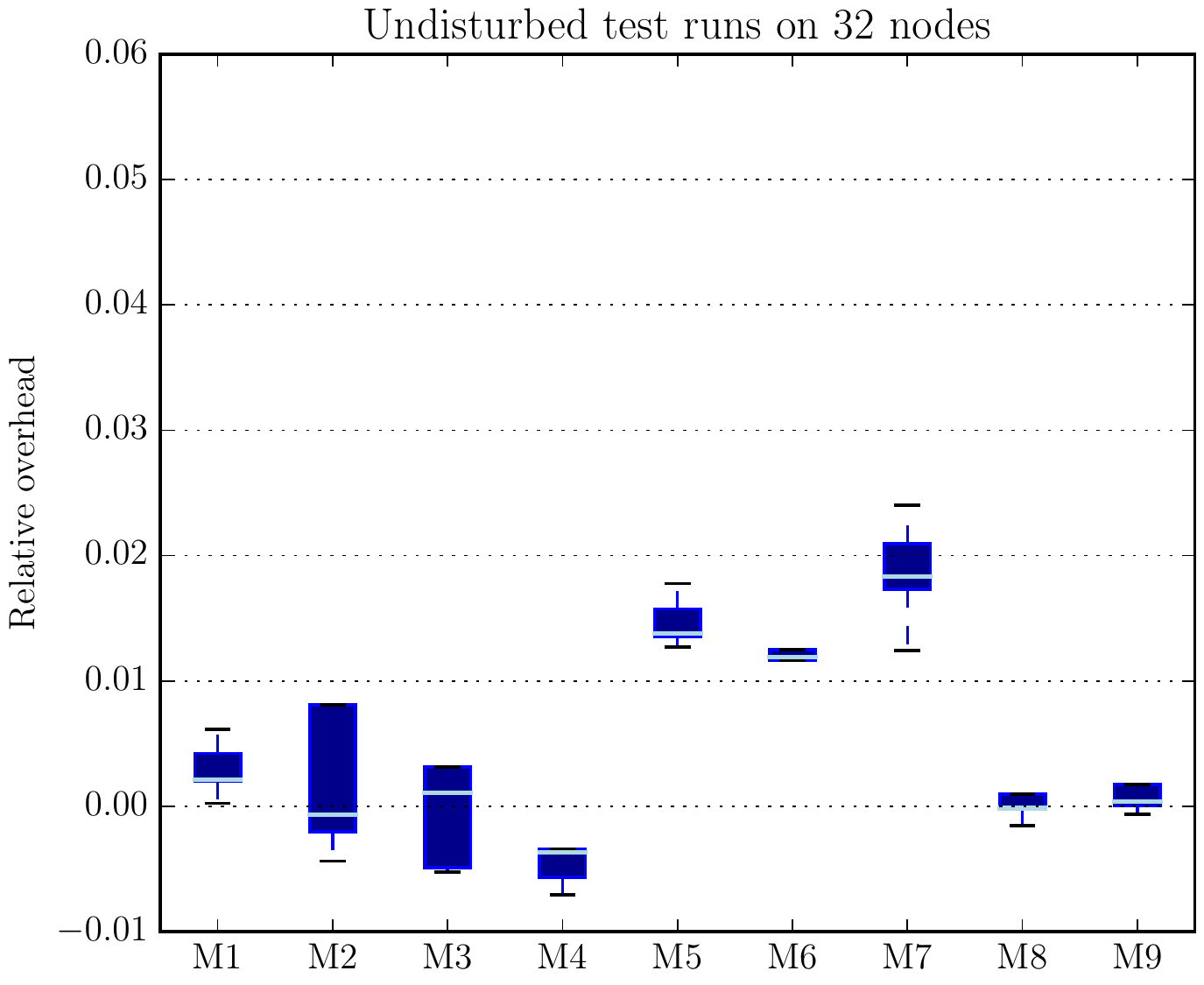}
\caption{Relative runtime overheads when sending the additional vector elements needed to achieve the desired redundancy, compared to the non-resilient case}
\label{fig:overheads:undisturbed}
\end{center}
\end{figure}

\begin{figure}[t]
\begin{center}
\includegraphics[width=\linewidth]{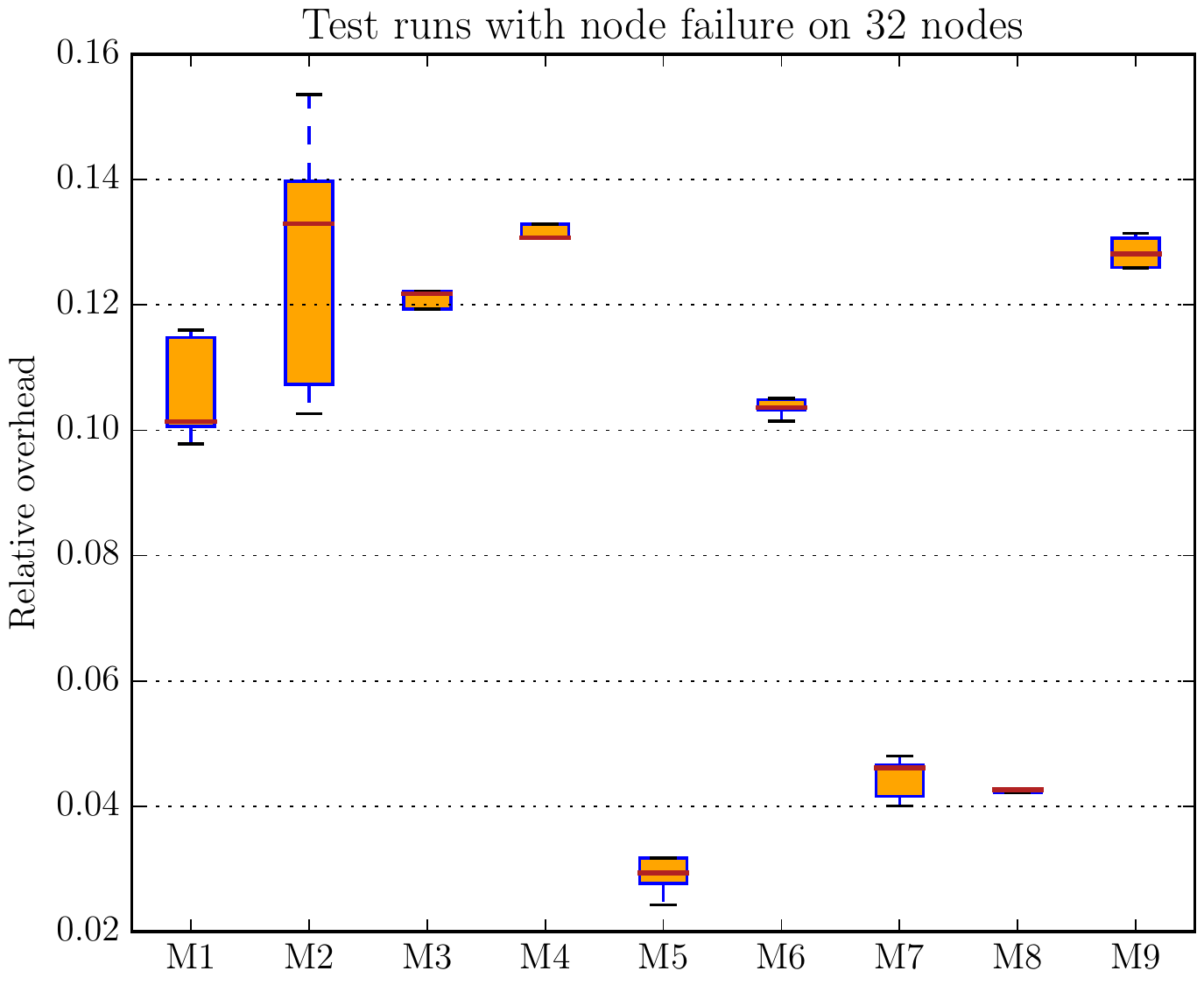}
\caption{Relative runtime overheads when simulating a node failure and reconstructing the state of the solver, compared to the non-resilient case}
\label{fig:overheads:failure}
\end{center}
\end{figure}

\begin{table}[t]
  \caption{Properties of the test matrices}
  \label{tab:experiments:matrices}
  \begin{center}
    \begin{tabu}{@{} l l S[table-format=7.0(0),table-number-alignment=left]
                 S[table-format=8.0(0),table-number-alignment=left] @{}}
      \toprule

      ID &
      Name &
      \multicolumn{1}{l}{Size \(n\)} &
      \multicolumn{1}{l}{Non-zeros} \\

      \midrule

      M1 & bcsstk18        &   11948 &   149090 \\
      M2 & s1rmt3m1        &    5489 &   217651 \\
      M3 & s1rmq4m1        &    5489 &   262411 \\
      M4 & bcsstk17        &   10974 &   428650 \\
      M5 & parabolic\_fem  &  525825 &  3674625 \\
      M6 & offshore        &  259789 &  4242673 \\
      M7 & G3\_circuit     & 1585478 &  7660826 \\
      M8 & Emilia\_923     &  923136 & 40373538 \\
      M9 & Hook\_1498      & 1498023 & 59374451 \\

      \bottomrule
    \end{tabu}
  \end{center}
\end{table}

\subsection{Implementation and experimental setup}

We implemented the parallel \ppcg algorithm in C, using the GNU Scientific
Library (GSL) to store data structures like vectors and matrices, and MPI for
parallelization.
In our experiments we used GSL 2.5, Intel MPI 2018 Update 4, and OpenBLAS 0.3.5
for BLAS operations with GSL.
We compiled with the Intel C compiler 18.0.5 with compiler flag \texttt{-O3}.

The convergence criterion for our solver is a reduction of the relative
residual norm by a factor of $10^{-8}$.
For solving the local linear systems during the reconstruction phase, we used a
factor of $10^{-11}$ as convergence criterion.
As suggested by Ghysels and Vanroose~\cite{Ghysels2014a}, our \ppcg
implementation provides the opportunity to perform residual replacement to
improve the accuracy of the result.
In all of our test runs, residual replacement was performed every 50 iterations
(this value is also used in~\cite{Ghysels2014a}).

Our experiments were executed on 32 nodes of the ``Hydra'' cluster situated at
TU Wien.
We used one process per node, which is sufficient to obtain representative
results for the reconstruction phase.

\begin{table*}[t]
  \caption{Experimental results}
  \label{tab:experiments:results}
  \begin{center}
    \begin{tabu}{@{} l
        S[table-format=2.2,table-number-alignment=left]
        S[table-format=5.0,table-number-alignment=left]
        S[table-format=+1.2,table-number-alignment=left]
        S[table-format=5.0,table-number-alignment=left]
        S[table-format=2.2,table-number-alignment=left] @{}}
      \toprule

			ID &
			\(t_0\) [s] &
			\multicolumn{1}{l}{Iterations until} &
			\multicolumn{1}{l}{Relative overhead with} &
			\multicolumn{1}{l}{Iterations until convergence} &
			\multicolumn{1}{l}{Relative overhead} \\

			&
			&
			\multicolumn{1}{l}{convergence} &
			\multicolumn{1}{l}{redundant copies [\%]} &
			\multicolumn{1}{l}{if node failure occurs} &
			\multicolumn{1}{l}{with node failure [\%]} \\

			\midrule

			M1 &
			0.09 & 
			1551 &
			0.21 &
			1550 &
			10.14 \\

			M2 &
			0.03 & 
			741 &
			-0.07 &
			742 &
			13.30 \\

			M3 &
			0.03 & 
			673 &
			0.11 &
			674 &
			12.18 \\

			M4 &
			0.16 & 
			2531 &
			-0.37 &
			2669 &
                        13.07 \\
      
                        M5 &
			4.61 & 
			2554 &
			1.38 &
			2554 &
                        2.94 \\

                        M6 &
			2.80 & 
			1948 &
			1.19 &
			1950 &
                        10.37 \\

                        M7 &
			15.39 & 
			3097 &
			1.83 &
			3097 &
                        4.61 \\

			M8 &
			47.18 & 
			10227 &
			-0.02 &
			10229 &
			4.26 \\

			M9 &
			35.46 & 
			4703 &
			0.04 &
			4703 &
			12.81 \\

      \bottomrule
    \end{tabu}
  \end{center}
\end{table*}

For each matrix, three different sets of test runs were executed.
The first set solves the linear system with the non-resilient standard \ppcg
algorithm.
In the second set, our strategy for guaranteed data redundancy
(cf.\ \Sec{redundancy}) is used, but no node failure occurs during the run.
Finally, in the last set, a node failure is simulated and the state of the
solver is reconstructed as described in \Sec{recovery:ppcg} and \Alg{esrppcg}.
Node failures are introduced after 50\% of the solver progress (i.e., after
50\% of the iterations the solver needs to converge for a particular matrix if
no failure occurs) and are always simulated at the node with rank 0.

All time measurements shown in the following are averaged over five test runs.
They only represent the time needed from the start of the iterative solver
until convergence, ignoring the time needed for setup operations such as
reading the matrix from a file or creating the preconditioner.
Similarly, for the test runs with reconstruction, the measured overheads
concern the time needed for the recovery of the lost vectors and scalars.
The reloading of the system matrix and the preconditioner on the replacement
node is excluded from the time measurements, since this would also be necessary
for any other approach (like checkpointing) and, therefore, does not provide
any information about the performance of our specific method.

\subsection{Results}

\Tab{experiments:results} summarizes the experimental results for our nine test
matrices.
The runtime for executing the non-resilient standard \ppcg solver is denoted as
\(t_0\).
The relative overheads with respect to \(t_0\) are listed in
\Tab{experiments:results} and visualized as boxplots in
\Figs{overheads:undisturbed}{overheads:failure}, showing the case without and
with node failures, respectively.
Note that the number of iterations until convergence marginally varies between
the two cases.
This is due to numerical effects during the reconstruction phase, which may
cause the reconstructed state of the solver to slightly deviate from the state
before the node failure, thus leading to a different subsequent behavior of the
solver.

We observe almost negligible relative overheads of well below 3\% for all our
test runs with additional data redundancy but without any node failures.
For the test runs with pure band matrices (M1--M4, M8, and M9), the overheads
even are within \(\pm 1\%\), which can be explained with system effects.
This indicates that overlapping the global dot product with the SpMV (including
the additional data redundancy) indeed works best for banded matrices.
Hence, our data redundancy strategy (as outlined in \Sec{redundancy}) can be
considered a particularly good fit for the \ppcg solver, which has been
primarily designed for band matrices (cf.\ \Sec{experiments:data}).

The relative overheads of approximately 3\% to 13\% for the test runs with node
failures are in a similar range as previous results for recovering from a node
failure in the context of the classical (non-pipelined) \pcg solver
\cite{Pachajoa2019a}.
This demonstrates the efficiency of our novel recovery algorithm for the \ppcg
method.

\section{Conclusions}
\label{sec:conclusions}

In this paper, we first reviewed three existing communication-hiding and thus
scalable variants of the \pcg and \pcr algorithms.
We then proposed an extension to these algorithms in order to make them
resilient against the potential failure of compute nodes \emph{without}
compromizing the scalability of the algorithms.
In fact, the improved resilience may even have positive effects on the
scalability for massively parallel systems.
Our experimental evaluation of the \ppcg algorithm illustrates that the
overheads caused by ensuring resilience against potential node failures and by
reconstructing the state of the solver after a node failure are very low:
almost negligible in the failure-free scenario and between 3\% and 13\% when a
node fails.
In future work, we want to experimentally investigate the behavior of our
resilient pipelined \pcg solvers on large-scale parallel systems.

\section*{Acknowledgments}
This work has been funded by the Vienna Science and Technology Fund (WWTF)
through project ICT15-113.

\bibliographystyle{siam}
\bibliography{Bibliography}

\balance

\end{document}